\documentclass[12pt,a4paper]{article}
\pdfoutput=1
\usepackage{fourier}
\usepackage[T1]{fontenc}
\usepackage[protrusion=true,expansion=true]{microtype}
\usepackage{geometry}
\usepackage{graphicx}
\usepackage{amsthm}
\usepackage[cmex10]{amsmath}
\usepackage{natbib}
\usepackage{amssymb}
\usepackage{dcolumn}
\usepackage{booktabs}
\usepackage{bm}
\usepackage{mathtools}
\usepackage{enumitem}
\usepackage[ruled,vlined]{algorithm2e}
\usepackage[format=hang, labelfont=bf,
tableposition=top, figureposition=bottom]{caption}
\usepackage{parskip}
\usepackage[title]{appendix}
\usepackage{fullpage}
\usepackage{fancyhdr}
\usepackage{setspace}
\usepackage{float}
\makeatletter
\def\thm@space@setup{%
	\thm@preskip=\parskip \thm@postskip=0pt
}
\makeatother
\newcolumntype{d}[1]{D{.}{\cdot}{#1}}
\numberwithin{equation}{section}
\mathtoolsset{showonlyrefs}
\theoremstyle{plain}
\newtheorem{theorem}{Theorem}[section]
\newtheorem{corollary}{Corollary}[theorem]

\newtheorem{example}{Example}

\title{\textbf{
		Exchanging Goods Using Valuable Money
}}

\author{J. V. Howard%
	\thanks{Department of Mathematics,
		London School of Economics, UK. \texttt{  j.v.howard@lse.ac.uk}}
}%
\date{11 September, 2023}
\begin{document}
	\maketitle
	\begin{abstract}
		A group of people wishes to use money to exchange goods efficiently over several time periods. However, there are disadvantages to using any of the goods as money, and in addition fiat money issued in the form of notes or coins will be valueless in the final time period, and hence in all earlier periods. Also, Walrasian market prices are determined only up to an arbitrary rescaling. Nevertheless we show that it is possible to devise a system which uses money to exchange goods and in which money has a determinate positive value. In this system, tokens are initially supplied to all traders by a central authority and recovered by a purchase tax. All trades must be made using tokens or promissory notes for tokens. This mechanism controls the flow as well as  the stock of money: it introduces some trading frictions, some redistribution of wealth, and some distortion of prices, but these effects can all be made small.
	\end{abstract}
\section{Introduction}
\label{sec:introduction}
A fundamental problem in economics is to describe how a group of people could exchange goods or services efficiently to their mutual benefit. Goods can be sold by auction, but that requires money. General equilibrium theory (Walras, Arrow-Debreu) shows how markets can work to exchange homogeneous and divisible goods, but the basic theory has no place for money. However, we observe in practice that money is extremely useful in decentralising economic transactions in both space and time. \cite{Hahn1965} asked whether it is possible to construct GE models where money has to have a positive value in any equilibrium solution (Hahn's problem).
\par
In essence, the problem arises because traders are prepared to exchange goods or services for money only because they believe that in the future they will be able to exchange the money for other goods or services. This implies that fiat money can work as a medium of exchange only when traders can use money to buy something of value from the government (so that money is cleared from the economy). Otherwise, if there were only a (known) finite number of time periods, money would be valueless on the last day, and hence on all days. Even when time is endless, there will still be equilibria in which the value of money is zero.
\par
Economists have devised models in which money has positive value (see Section~\ref{sec:related}), but in this paper we will take an alternative approach. We will think of the problem as one of mechanism design or economic engineering, and so we will try to design an economic system in which money is used in all trades, and where the value of money can be controlled by society from one time period to another. We believe that even a single very simplified model showing how such a system could work (proof of concept) is of theoretical interest.
\par
So we start by defining such a system for a pure exchange economy. The mechanism requires that in every period all traders are supplied by the central bank with `gold standard' money in the form of tokens. They then trade using either tokens or promissory notes for tokens. All trade must be an exchange of money for goods: barter is not allowed. All trades of money for goods incur a fixed percentage purchase tax which must be paid in tokens. The number of tokens supplied each period increases slightly over time, so that there is a continual small but controlled inflation, which discourages traders from hoarding tokens. This, in brief, defines the system.
\par
Our claim is that in such an economy it is possible to calculate not just what goods would be traded and in what quantities, but also at what prices. To do this, we do not need any new theory, we can just analyse any particular well-defined economy using standard Walrasian general equilibrium theory, with small modifications to allow for the effects of the purchase tax. Our overall conclusion is that if the purchase tax is small, the new system will have an equilibrium close to an equilibrium of the original exchange economy, but with determinate price levels in tokens. (We will also demonstrate that extreme cases can occur in which there is no equilibrium.)
\par
It might be argued that traders could just avoid using the system by bartering. However, in practice it seems that people are extremely wedded to trading using money -- even for black market or criminal transactions. They might also try to avoid paying the purchase tax, but tax evasion is a problem that countries are used to grappling with: no doubt some evasion will always take place, but its scope may be limited.
\par
Since we are not presenting any new theory (except possibly for showing how to calculate the effects of a purchase tax on trading volumes), we will proceed by an extended numerical example. Although very simple (two goods, three traders), it is difficult to calculate the solutions without the use of a computer. The algorithm we use is described in the Appendix. We suggest that anyone wishing to check our calculations should use a program such as Mathematica, Maple, or R.
\par
In order to give more motivation, we present the example in the form of a narrative about some desert island economies. The islanders consider various proposed systems, but one island tries the mechanism described above. We also show how the standard Edgeworth box diagram can be modified to analyse the equilibrium when there are only two goods and two traders. Finally we briefly survey some related work.
\section{An imaginary commodity}
\label{sec:example}
One possible story would be about a group of families who help each other with baby-sitting, taking the children to school, car sharing for commuting, and so on. Some people begin to suspect that others are wilfully exploiting the community by taking out more than they contribute. So the idea is proposed of introducing some sort of local currency within the group to keep track of debts and credits. But this scenario is still not sufficiently well defined for our purposes. We are looking for an even simpler (if less realistic) setting.
\par
So let us imagine that after a shipwreck a group of people is stranded on a desert island. Fortunately everyone is unharmed, and all escape from the ship with a box of perishable food. They establish contact with the rest of the world, who inform them that they will be rescued by boat in one or two week's time. If they have to stay a second week, fresh food boxes will be dropped by parachute. Meanwhile they are invited to enjoy their enforced holiday, and share their food.
\par
However, they have all acquired a strong feeling of property rights over their own food boxes. Some have a good mixture of items, but others have larger quantities of only one or two fruits or vegetables. As people's tastes and appetites differ, there are obvious benefits to trade, but they have no common currency. They decide it would be an interesting exercise to devise a system for exchanging the foodstuffs, a system which will start and end on the island, and leave no outstanding debts when they leave.
\par
They reject almost immediately the idea of looking for a rare object (perhaps an unusual shell) to be used as money: there would be too much danger that someone might find a place where the shells were plentiful, thereby becoming an instant millionaire and also debasing the currency. Nominating one foodstuff to be used as the medium of exchange is considered more seriously, but there seems no suitable choice. Giving everyone a starting stock of centrally produced banknotes would not work: nobody would want to be left holding the notes at the end of the week (when they would be worthless), so prices would be bid up indefinitely in the markets. However, they know there is another way of creating money: debt.
\par
They understand that it should be possible for the traders themselves to create money by writing promissory notes. But IOU notes for what? Many countries once traded using IOU notes for gold or silver. But, since leaving the gold standard, countries like the US and the UK effectively issue notes for an imaginary commodity (dollars or pounds). They decide to follow this route.
\par
So the newly created finance committee announces that the island's currency unit will be the crown, and that they expect a crown to be about the price of an apple or a pear. They plan to invite those who wish to trade to take their goods to market, where some sort of Walrasian auction process will take place, leading to suggested prices for all goods. Traders will then be asked to create money by writing and signing IOU notes for one crown. Traders can keep their possessions secret, but the number of notes they have created will be public knowledge. Well known traders in good standing will have their notes accepted by everyone, and so paper money will be created. (In Scotland three Scottish banks as well as the Bank of England can issue banknotes.) In fact we will assume that all traders are in good standing, and so everyone can create money. Trading will then take place using these notes as money, until everyone has the goods they want plus a bundle of other people's IOU notes.
\par
At times during the trading period a person may have accumulated more debts than credits, but he must ensure that by the close of the period he can match his debts (his IOU notes) with credits (other people's notes). In fact, everybody must arrange to finish trading with their transactions exactly balanced. Then, at the close of trade, outstanding debts must be reconciled.
\par
So a trader with an IOU note will visit the issuer of the note and demand payment. In general the issuer will take back and destroy her own note and give in exchange another note issued by someone else, but eventually the trader should receive one of his own notes, which he must then destroy. Assuming that everyone trades honestly, all debts will eventually clear, and all the money created to facilitate trade will be destroyed.
\par
However, this system has vulnerabilities. In Britain, if you take a \pounds 20 note into a bank, and ask them to honour the pledge written on it (``I promise to pay the bearer on demand the sum of twenty pounds''), you are likely at best to be given another \pounds 20 note (or 2 ten pound notes, or 20  one pound coins, and so on). Correspondingly, what if a trader pays off his debt by simply writing a fresh IOU note. The finance committee could make a rule saying that debts must be paid using another person's notes, but then there could be scams involving two traders acting in collusion. The committee decides that the first day will be allowed for trading, and that after that day nobody can issue more notes. The process of resolving debt is then grounded.
\par
We might describe this simple system as ``money as debt for an imaginary commodity''. It seems workable (with suitable penalties for cheating), but it determines only relative prices. (If the prices of all goods were doubled, there would be an equally good equilibrium, probably requiring the creation of about twice as much money.) We could try to fix prices by demanding that the IOU notes must be written for one particular good (e.g. apples), but when we come to look at two period problems, there may not be any apples in the second period, while we want the general value of money to remain much the same.
\par
In the extreme version of this scheme, everyone starts by creating as much money as he needs for all his purchases. He then trades, buying goods with his money and receiving other people's money for his sales. Each note is used exactly once on day one, and all notes are destroyed on day two.
\par
We now introduce our simple numerical example, which we will modify later when we see how we might give money a determinate value.
\begin{example}
	\label{ex:original}
	We suppose that there are equal large numbers of three types of trader $A$ (the rich), $C$ (the poor), and $B$ (the middle class), and just two goods (apples and pears). Type $A$ starts with $90$ apples and $30$ pears. Type $C$ starts with $10$ apples and $70$ pears. Type $B$ starts with $50$ of each. They all have the same utility function of their holdings of apples, $x$, and pears, $y$.
	\begin{equation*}
		u(x, y) = \sqrt{x} + \sqrt{y} \text{ .}
	\end{equation*}
\end{example}
So the rich start with a utility of $15.0$, the poor with $11.5$, and the middle class with $14.1$. In a market equilibrium, apples and pears will have the same price, and $A$ (types) will buy $30$ pears from $C$ (types) and sell $30$ apples to them, until everyone has equalised their holdings of the two fruits. $A$'s utility will then be $15.5$, $C$'s will be $12.6$, and $B$ will remain at $14.1$. Figure~\ref{fig:original} shows these movements in quantity space. Traders $A$ and $C$ start at a blue square and move along their green budget constraint to maximise their utility at a red circle.
\begin{figure}[H]
	\centering
	\includegraphics[scale=1.0, clip]{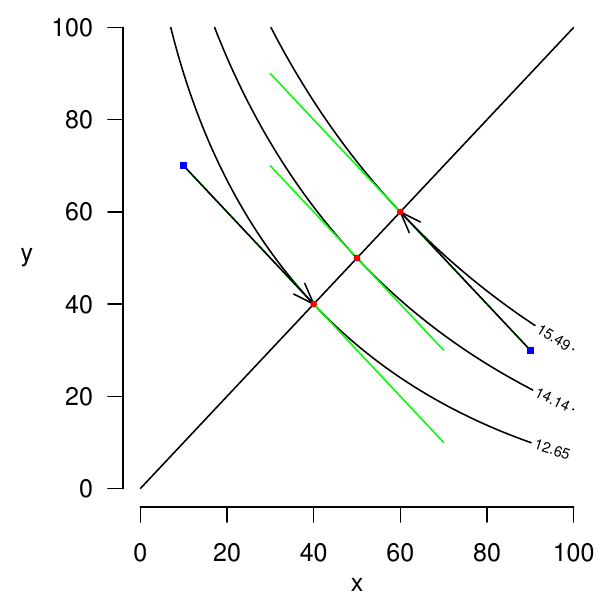}
	\caption{Movements to market equilibrium}
	\label{fig:original}
\end{figure}
\par
Let the prices of apples and pears be $p$ and $q$ respectively. Figure~\ref{fig:originalxs} shows contours of the excess demand for apples in price space. (Blue contours, with slope greater than $1$,  have positive excess demand; red contours, with slope less than $1$, negative.) On the green diagonal $p = q$ the supply of apples equals the demand. If we plotted the excess demand for pears, the contour of zero excess demand would also be $p = q$. The two equations in $p$ and $q$ determine only the ratio $p:q$.
\begin{figure}[H]
	\centering
	\includegraphics[scale=1.0 , clip]{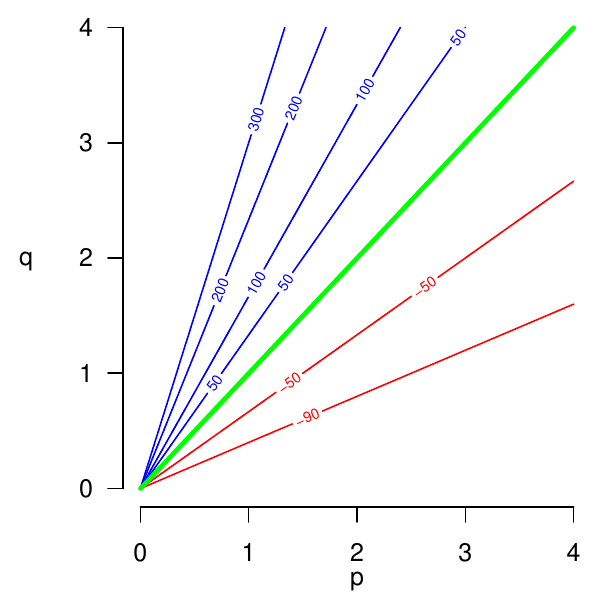}
	\caption{Excess demand contours}
	\label{fig:originalxs}
\end{figure}
\par
So we have a simple system based on traders creating their own money, IOU notes for crowns. There are no real crowns. Apart from the initial suggestion for the approximate value of a crown, the finance committee has no further influence on the general price level. We have described a Walrasian pure exchange economy, and general equilibrium theory tells us that at least one equilibrium will exist under simple assumptions (which our model satisfies).
\section{Real money}
\label{sec:tokens}
\par
Suppose now that the finance committee on another (almost identical) desert island prefers not to use the system described in Section~\ref{sec:example}, because of the difficulty of maintaining the value of money over several periods, but instead tries the system outlined in Section~\ref{sec:introduction}. It sets up a central bank which then gives each trader a fixed number ($n$) of real crowns. These real crowns are to be the tokens referred to in Section~\ref{sec:introduction}, and we will refer to crowns rather than tokens in our examples. Whenever the trader spends a crown on food, he most return $r$ crowns to the bank ($r < 1$). On the first day all trades must be made using real crowns or IOU notes for crowns, and the central bank can similarly be paid in real money or promissory notes. On day two (settlement day) debts are reconciled as before, except that debts to the central bank must be paid in real crowns.
\begin{example}
	Continuing with the situation in Example~\ref{ex:original}, suppose now that before trading all traders are given $n = 6$ crowns, and instructed to return one crown to the bank for each five crowns spent purchasing goods (i.e. $r = 1/5$).
\end{example}
\par
Suppose that the price of each fruit is (say) $2$ crowns. (These are not equilibrium prices.) Consider a type $A$ trader who starts with $90$ apples, $30$ pears, and $6$ real crowns. She could consider buying more fruit.  Because of the $20\%$ purchase tax, she sees a price of $2.4$ rather than $2$, so she can buy $5/2$ fruits, paying $5$ crowns and keeping $1$ crown to pay the bank. The green line in Figure~\ref{fig:tokens} shows the possibilities using this strategy.
\begin{figure}[H]
	\centering
	\includegraphics[scale=1.0 , clip]{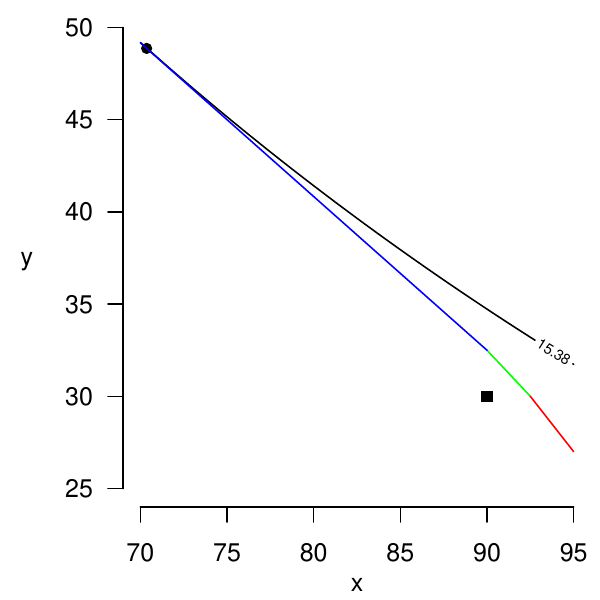}
	\caption{Budget constraints when $r = 20\%$}
	\label{fig:tokens}
\end{figure}
\par
She next considers selling $- \delta x$ apples and buying as many pears as she can with the proceeds and her $6$ crowns. After the sale she will have capital worth $6 - 2 \delta x$ crowns, so she can buy $\delta y = (6 - 2 \delta x)/2.4$ pears. Her possible trades must satisfy $2 \delta x + 2.4 \delta y = 6$. So if she finishes with holdings of $(x, y)$ we must have that
\begin{align*}
2(x - 90) + 2.4(y - 30) = 6 \\
2x + 2.4y = 258 \text{ .}
\end{align*}
The blue line in Figure~\ref{fig:tokens} shows these possibilities.
\par
Similarly, she could sell $- \delta y$ pears and then buy as many apples as she could. After the sale she would have capital worth $6 - 2 \delta y$ crowns, so she could buy $ \delta x = (6 - 2 \delta y)/2.4$ pears. Her possible trades must satisfy $2.4 \delta x + 2 \delta y = 6$. So if she planned to finish with holdings of $(x, y)$ we must have that
\begin{align*}
2.4(x - 90) + 2(y - 30) = 6 \\
2.4x + 2y = 282 \text{ .}
\end{align*}
These possible outcomes lie on the red line in Figure~\ref{fig:tokens}.
\par
In the general case, if a trader purchases $ \delta x_{i}$ of good $i$ ($i = 1, \ldots, k$) at unit price $p_{i}$, with no initial capital and no purchase tax her budget constraint would be
\begin{equation*}
	\sum_{i} p_{i} \delta x_{i} \le 0 \text{ .}
\end{equation*}
After giving her capital of $n$ crowns and introducing purchase tax at rate $r$, she will have $2^{k} - 1$ constraints, one for each non-empty set $S \subseteq \{1, \ldots, k\}$, of the form
\begin{equation*}
	\sum_{i} p_{i} \delta x_{i}\left(1 + \phi_{i} r\right) \le n \text{ ;}
\end{equation*}
where $\phi_{i}$ is an indicator variable for membership of $S$ (the set of goods the trader wishes to buy). Essentially, she must pay more to buy a good than she receives from selling it: trade is not frictionless. Since the budget set (assuming free disposal) is convex and her preferences are strictly convex, there will be a unique optimum point for the trader at the given prices. It may be a corner point: for a range of prices, an $A$ trader may wish to use all his crowns to buy more pears, but not wish to sell any apples to buy yet more pears. 
\par
The utility contour shown in the figure is the highest Player $A$ can attain at the assumed prices: she achieves it by moving to the point marked with a circle . In fact, straightforward calculations show that she will sell $19.64$ apples and buy $18.86$ pears. She receives $39.27$ crowns for the apples, and she started with $6$ crowns. The pears cost her $37.73$ crowns (paid to the seller), and she must also give $7.55$ crowns to the bank. Her budget is balanced. Her demands for apples and pears are shown in the second column of Table~\ref{tab:demandunits}. Her balance sheet is shown in the corresponding column of Table~\ref{tab:demandgroats}.
\begin{table}[H]
	\centering
	\caption{Demand (units) for goods by trader}
	\begin{tabular}{d{3.2}d{3.2}d{3.2}d{3.2}d{3.2}}
		\toprule%
		& \multicolumn{3}{c}{Trader}\\
		\cmidrule{2 - 4}
		\multicolumn{1}{c}{Good} & \multicolumn{1}{c}{$A$} & \multicolumn{1}{c}{$B$} & \multicolumn{1}{c}{$C$} & \multicolumn{1}{c}{Total} \\
		\midrule
		\multicolumn{1}{r}{Apples} & -19.64 & 1.25 & 22.20 & 3.81\\
		\multicolumn{1}{r}{Pears} & 18.86 & 1.25 & -23.64 & -3.52\\
		\bottomrule
	\end{tabular}
	\label{tab:demandunits}
\end{table}
\begin{table}[H]
	\centering
	\caption{Balance sheets of traders}
	\begin{tabular}{d{3.2}d{3.2}d{3.2}d{3.2}d{3.2}}
		\toprule%
		& \multicolumn{3}{c}{Trader}& \\
		\cmidrule{2 - 4}%
		\multicolumn{1}{c}{Good} & \multicolumn{1}{c}{$A$} & \multicolumn{1}{c}{$B$} & \multicolumn{1}{c}{$C$} & \multicolumn{1}{c}{Total} \\
		\midrule%
		\multicolumn{1}{r}{Initial} & 6.00 & 6.00 & 6.00 & 18.00\\
		\midrule%
		\multicolumn{1}{r}{Apples} & 39.27 & -2.50 & -44.39 & -7.62\\
		\multicolumn{1}{r}{Pears} & -37.73 & -2.50 & 47.27 & 7.05\\
		\multicolumn{1}{r}{Tax} & -7.55 & -1.00 & -8.88 & -17.42\\
		\midrule%
		\multicolumn{1}{c}{Final} & 0 & 0 & 0 & 0\\
		\bottomrule
	\end{tabular}
	\label{tab:demandgroats}
\end{table}
\par
In a similar way, $C$ will sell $23.64$ pears and buy $22.20$ apples and pay $8.88$ tax to the bank.
\par
$B$ will clearly use his $6$ crowns to buy $3/2.4 = 1.25$ apples and $1.25$ pears. He pays $5$ crowns for the fruit and gives $1$ crown to the bank.
\par
We now see that at the assumed prices, the total demand for apples is $22.20 + 1.25 = 23.45$ The total supply is $19.64$. The price of apples will tend to rise. The supply of pears exceeds the demand: the price of pears will tend to fall. But we now have a third good: real crowns (needed to pay purchase tax to the central bank). The supply of real crowns is clearly 18, and Table 2 shows that the demand is 17.42. There is an oversupply of real money, and the general price level will tend to rise. This means that we may expect that the final price for apples will be higher than 2, but remain uncertain about the final price for pears.
\par
We look for an equilibrium keeping the purchase tax rate, $r$, at $1/5$. We find that we must have $p = 2.075$ and $q = 2.022$. (The Appendix describes how to calculate these prices.)
\par
We can then calculate the revised balance sheets shown in Table~\ref{tab:eqdemand}.
\begin{table}[H]
	\centering
	\caption{Equilibrium balance sheets of traders}
	\begin{tabular}{d{3.2}d{3.2}d{3.2}d{3.2}d{3.2}}
		\toprule%
		& \multicolumn{3}{c}{Trader}& \\
		\cmidrule{2 - 4}
		\multicolumn{1}{c}{Good} & \multicolumn{1}{c}{$A$} & \multicolumn{1}{c}{$B$} & \multicolumn{1}{c}{$C$} & \multicolumn{1}{c}{Total} \\
		\midrule%
		\multicolumn{1}{r}{Initial} & 6.00 & 6.00 & 6.00 & 18.00\\
		\midrule%
		\multicolumn{1}{r}{Apples} & 43.64 & 0.00 & -43.64 & 0\\
		\multicolumn{1}{r}{Pears} & -41.36 & -5.00 & 46.36 & 0\\
		\multicolumn{1}{r}{Tax} & -8.27 & -1.00 & -8.73 & -18.00\\
		\midrule%
		\multicolumn{1}{c}{Final} & 0 & 0 & 0 & 0\\
		\bottomrule
	\end{tabular}
	\label{tab:eqdemand}
\end{table}
In this solution, trader $B$ maximises his utility at a corner point of his budget set (he buys only  pears). We can recalculate Figure~\ref{fig:original}. ($B$'s movement is too small to show with an arrow.)
\begin{figure}[H]
	\centering
	\includegraphics[scale=1.0, clip]{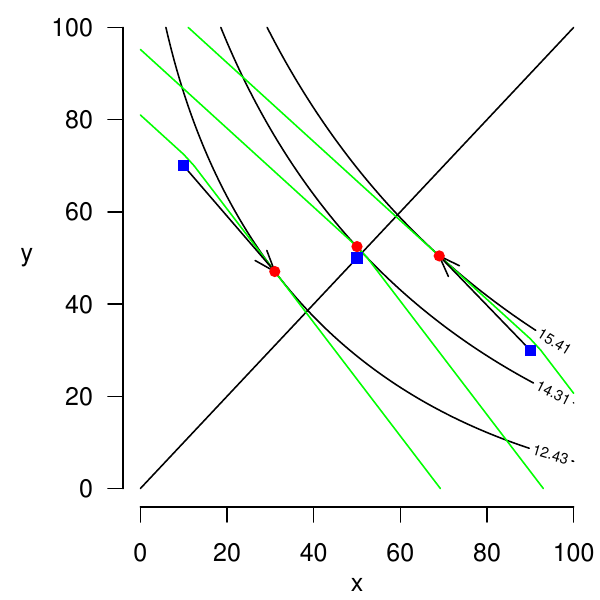}
	\caption{Revised movements to market equilibrium}
	\label{fig:reveqmoves}
\end{figure}
\par
If we look at the contour lines where supply equals demand for the three goods, we find Figure~\ref{fig:threecontours}. The green line shows where the apple market is in equilibrium and the red line is for pears. In contrast to Figure~\ref{fig:originalxs}, these two lines are not coincident, and we have a unique (stable) equilibrium. The irregular blue oval shows where demand and supply for real crowns are equal. All three lines intersect at the same point (Walras' Law).
\begin{figure}[H]
	\centering
	\includegraphics[scale=1.0, clip]{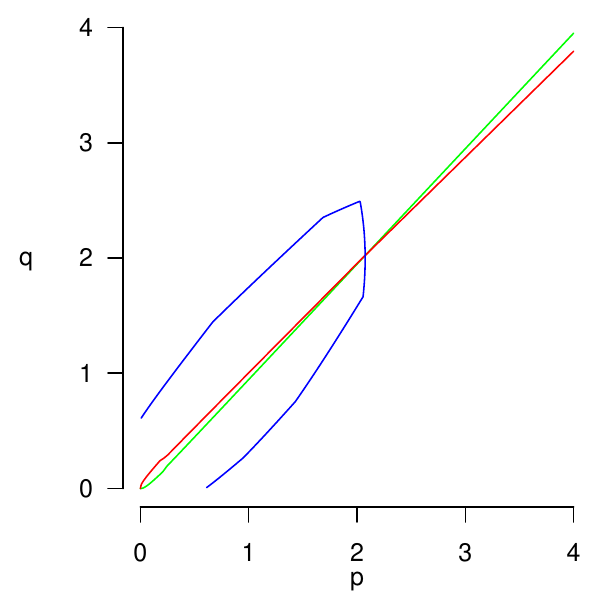}
	\caption{Equilibrium for three markets when $n = 6$ and $r = 1/5$}
	\label{fig:threecontours}
\end{figure}
Note that in order to calculate these three contour lines, we have to calculate, for each pair of $p$ and $q$ values, which of the three lines and two corner points in Figure~\ref{fig:tokens} is optimal for each of the three traders. (The Appendix describes how to do this.) The best choice for a trader changes as $p$ and $q$ vary, and this accounts for the lack of smoothness of the three lines.
\par
In this system, effectively the finance committee and the central bank impose a purchase tax and then divide the proceeds equally amongst the citizens -- although in fact the bank pays out the proceeds before collecting the tax. In general, the tax will have three effects:
\begin{enumerate}[label=(\alph*), topsep=0.0\baselineskip]
	\item it is somewhat redistributive as everyone receives the same income from the bank;
	\item people who wish to trade less than average are at an advantage because they have surplus tokens to sell;
	\item all buyers face higher prices, so trade will be restricted (as with any purchase tax).
\end{enumerate}
\par
We could of course have specified a sales tax (paid by the vendor) instead of a purchase tax (paid by the buyer), but this would essentially be exactly the same system. The sellers would add the tax they had to pay to the price, so we would get the same outcome from a tax rate of $r/(1 + r)$ on sales as of $r$ on purchases.
\par
We can think of the first system (without real crowns) as the limiting case when $n = r = 0$. In our example using the first system, apples and pears are equally priced, so we can measure wealth by total holdings of fruit. Table~\ref{tab:wealth} shows how wealth is changed by the system. $A$ and $C$ lose because they wish to trade more than the average, and $B$ gains because he does not want to trade. In the example, this effect is dominant.
\begin{table}[H]
	\centering
	\caption{Total holdings of goods}
	\begin{tabular}{d{4.1}d{4.1}d{4.1}d{4.1}}
		\toprule
		& \multicolumn{3}{c}{Trader} \\
		\cmidrule{2 - 4}
		\multicolumn{1}{c}{Situation} & \multicolumn{1}{c}{$A$} & \multicolumn{1}{c}{$B$} & \multicolumn{1}{c}{$C$} \\
		\midrule%
		\multicolumn{1}{r}{At start} & 120.0 & 100.0 & 80.0 \\
		\multicolumn{1}{r}{At end when $n = r = 0$} & 120.0 & 100.0 & 80.0 \\
		\multicolumn{1}{r}{At end when $n = 6$ and $r = 20\%$} & 119.4 & 102.5 & 78.1 \\
		\bottomrule
	\end{tabular}
	\label{tab:wealth}
\end{table}
Of course the theory assumes that traders care about utility rather than wealth. Table~\ref{tab:utility} shows how trade increases utility.
\begin{table}[H]
	\centering
	\caption{Utility changes}
	\begin{tabular}{d{3.3}d{3.2}d{3.2}d{3.2}}
		\toprule
		& \multicolumn{3}{c}{Trader} \\
		\cmidrule{2 - 4}
		\multicolumn{1}{c}{Situation} & \multicolumn{1}{c}{$A$} & \multicolumn{1}{c}{$B$} & \multicolumn{1}{c}{$C$} \\
		\midrule
		\multicolumn{1}{r}{At start} & 14.96 & 14.14 & 11.53 \\
		\multicolumn{1}{r}{At end when $n = r = 0$} & 15.49 & 14.14 & 12.65 \\
		\multicolumn{1}{r}{At end when $n = 6$ and $r = 20\%$} & 15.41 & 14.31 & 12.43 \\
		\bottomrule
	\end{tabular}
	\label{tab:utility}
\end{table}
\par
This new system is based on real money (crowns) and IOU notes for crowns. We have moved from ``money as debt for an imaginary commodity'' to ``money as debt for real crowns''. Instead of a gold standard, we have a crown standard. Like gold, we may assume that crowns are kept in a secure vault (or some modern cyberspace version), and all trading on the first day is done using promissory notes. The bank also receives its purchase tax in the form of notes on day one. On day two, all credits and debts are reconciled, and the bank gets back all its crowns from the vaults. The finance committee now has two parameters under its control: $r$, the purchase tax rate, and $n$, the initial endowment of crowns for each trader. (In our example $r = 1/5$ and $n = 6$.) Equilibrium prices (if they exist) will be directly proportional to $n$, so  in fact we have a quantity theory of money. It seems this system can control prices, but will there always be an equilibrium solution?
\subsection{Existence of equilibrium}
\label{subsec:existence}
Suppose there is a market equilibrium without using real money (i.e. when $n = r = 0$). Will there always be an equilibrium when $n$ and $r$ are given definite positive values? The answer is clearly no. Modify our example so that the traders all start with equal amounts of the two goods (so $A$ might start with $60$ apples and $60$ pears, $B$ with $50$ of each, and $C$ with $40$ of each). There is an equilibrium with equal prices, $p = q$, in which no one wishes to trade. But as soon as real crowns are introduced prices will be bid up indefinitely. There is no equilibrium short of $p = q = \infty$. If we keep $n$ and $r$ fixed, we can find various starting positions for the three traders near the main diagonal $y = x$ where there will be no equilibrium.
\par
However, if we start with a situation where there is an equilibrium involving some trading when $n = r = 0$, then under reasonable continuity assumptions we would expect to have a similar equilibrium for fixed $n$ and sufficiently small $r > 0$. Also, for sufficiently small $r$, we should be able to make price distortions and redistribution effects as small as we like. In this sense, the proposed mechanism does formally solve Hahn's problem.
\subsection{The need for debt}
\label{subsec:needdebt}
Now we have a system with real money (real crowns), do we still need extra money (in the form of debt as promissory notes)? We could imagine everyone starts with an account with a balance of $6$ crowns, and a debit card allowing payments whilst the account is in credit. Then, in our example with $r = 1/5$, trader $B$ can simply use his debit card firstly to buy pears from $C$ using $5$ of his crowns, and then to pay the remaining crown as tax to the bank. Then he has completed all his trades and his account is empty. $C$ now has $11$ crowns.
\par
But life is not so simple for $A$ and $C$. Suppose $A$ uses $5$ of his crowns to buy pears from $C$, and then pays the remaining crown as tax. $A$ now has no crowns and $C$ has $16$. Next $C$ can use all his crowns to buy apples from $A$ (paying the appropriate amount of tax): then $A$ can buy more pears from $C$, and so on. After an infinite number of transactions, $A$ will have spent $455/11$ crowns buying pears from $C$ (and paid $91/11$ crowns in tax). $C$ will have spent $480/11$ crowns buying apples from $A$ (and paid $96/11$ crowns in tax). Eventually the bank gets $17$ crowns back from $A$ and $C$ together, and already has $1$ from $B$, so all the original $18$ crowns are returned, and no account is ever in debt.
\par
However, if the traders are issued with credit cards instead of debit cards, after $B$ has bought pears from $C$, $A$ can use credit to spend $455/11$ crowns on pears from $C$ and to pay $91/11$ crowns in tax. She then has debt of $480/11$ (approximately $44$) crowns. Finally $C$ (who now has $576/11$ crowns) spends it all on apples from $A$ and tax. All accounts are now zero.
\par
We see that in order to have enough liquidity in the system, it is still essential to allow traders to create money in the form of debt.
\subsection{Taxation and public spending}
\label{subsec:taxation}
So far we have envisaged a central bank but no government spending. (The redistributive effect of the token system should be thought of as a by-product of the need to maintain a stable currency, not as a social policy.) However, suppose a government is elected which wishes to tax and spend. Imagine, for example, that there are are food shortages on a neighbouring island, and the government decides to raise a poll tax of $12$ crowns and spend all the tax revenue on apples to be given to the neighbour.
\par
We can think of this as an extension to a market with $4$ traders, $A$, $B$, and $C$ starting with $6$ crowns and a debt of $12$ crowns, and the government starting with credit of $36$ crowns. Now, for some prices, a trader might not be able to clear his debt to the government, even if he sells all his goods. However, in this example we find equilibrium prices greater than 1, which does allow everyone to clear all their debts. So, the analysis proceeds as before except that buying both apples and pears is no longer a possible strategy for a trader. In fact it must be replaced by the strategy of selling at least one good in sufficient quantities to clear the debt to the government. With this change, we find that we must have $p = 1.661$ and $q = 1.466$.
\par
We then find the revised demands for goods shown in Table~\ref{tab:taxes01}. Note that the government must pay purchase taxes to the central bank, just like any other trader. The government can redistribute wealth amongst the citizens, but it has no power over the central bank which sets $n$ and $r$.
\begin{table}[H]
	\centering 
	\caption{Equilibrium balance sheets with taxation}
	\begin{tabular}{d{3.2}d{3.2}d{3.2}d{3.2}d{3.2}d{3.2}}
		\toprule
		& \multicolumn{4}{c}{Trader}& \\
		\cmidrule{2 - 5}
		\multicolumn{1}{c}{Good} & \multicolumn{1}{c}{$A$} & \multicolumn{1}{c}{$B$} & \multicolumn{1}{c}{$C$} & \multicolumn{1}{c}{$G$} & \multicolumn{1}{c}{Total} \\
		\midrule%
		\multicolumn{1}{r}{Initial} & 6.00 & 6.00 & 6.00 & 0.00 & 18.00\\
		\midrule
		\multicolumn{1}{r}{Apples} & 48.55 & 6.00 & -24.55 & -30.00 & 0\\
		\multicolumn{1}{r}{Pears} & -35.45 & 0.00 & 35.45 & 0.00 & 0\\
		\multicolumn{1}{r}{Poll tax} & -12.00 & -12.00 & -12.00 & 36.00 & 0\\
		\multicolumn{1}{r}{Purchase tax} & -7.09 & 0.00 & -4.91 & -6.00 & -18.00\\
		\midrule
		\multicolumn{1}{c}{Final} & 0 & 0 & 0 & 0 & 0\\
		\bottomrule
	\end{tabular}
	\label{tab:taxes01}
\end{table}
In this solution, trader $B$ still maximises his utility at a corner point of his budget set (he sells only apples). We can recalculate Figure~\ref{fig:original}.
\begin{figure}[H]
	\centering
	\includegraphics[scale=1.0, clip]{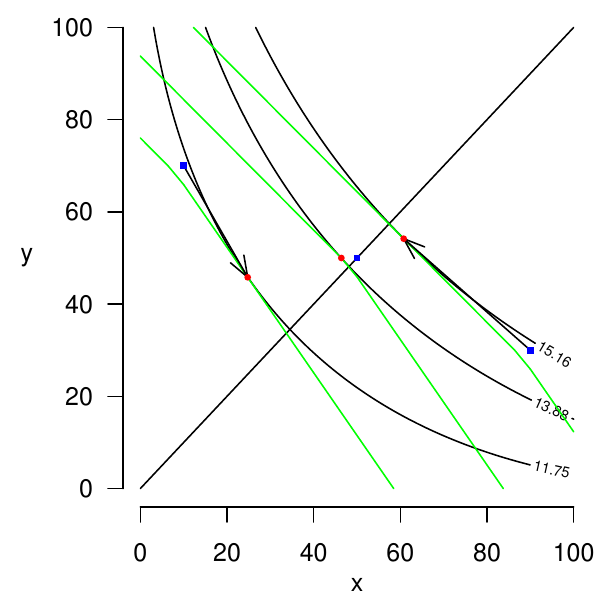}
	\caption{Movements to market equilibrium with taxation}
	\label{fig:reveqtaxmoves}
\end{figure}
\section{The Edgeworth box}
\label{sec:edgeworth}
Let us modify our example by removing trader $B$. With just two traders and two commodities, if we do not use real money we can represent the market in an Edgeworth box. Figure~\ref{fig:originaledge} shows the utility contours for the two traders (the second trader uses rotated axes). The contract curve (the locus of Pareto efficient points, here the line $y = x$) is the curve where the two contours have a common tangent. These tangents are shown as green lines. (Some authors restrict the contract curve to efficient points which are better than the initial position for both traders.) A point on the contract curve such that the common tangent line passes through the original holdings for the two traders is an equilibrium solution with the relative prices shown by the gradient of the tangent.
\begin{figure}[H]
	\centering
	\includegraphics[scale=1.0, clip]{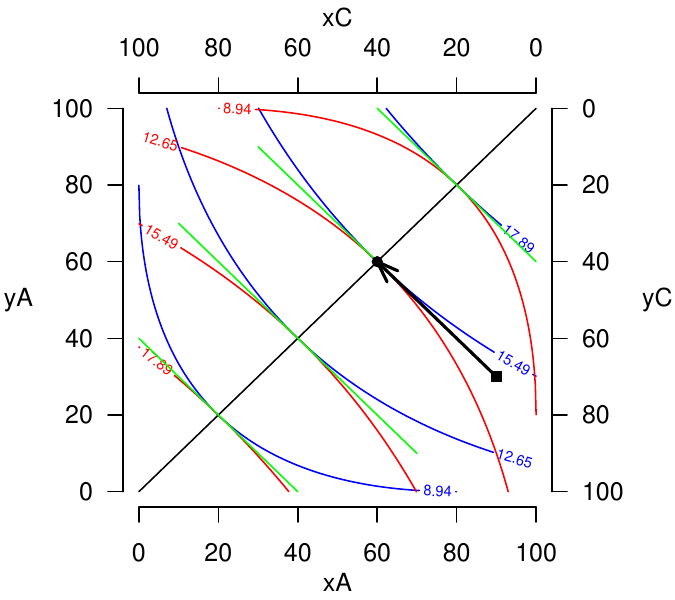}
	\caption{Edgeworth box} 
	\label{fig:originaledge}
\end{figure}
\par
How does this diagram change when we use tokens with the given values for $r$ ($20\%$) and $n$ ($6$)? We saw that each trader will have three budget constraints: Figure~\ref{fig:tokens} shows these for trader $A$, but clearly with just two traders we can find a solution only if $A$ moves to somewhere along the upper blue line and correspondingly $C$ to somewhere on the lower red constraint. If the prices of the goods are $p$ and $q$, the blue line has slope $p/(q(1 + r))$ and the red line $p(1 + r)/q$. The ratio of the slopes is $(1 + r)^{2}$. So for a point to be a possible solution, the slopes of the two utility curves at that point must have this ratio. So we replace the contract curve with a revised curve connecting points having this property, as shown in Figure~\ref{fig:tokensedge}.
\begin{figure}[H]
	\centering
	\includegraphics[scale=1.0, clip]{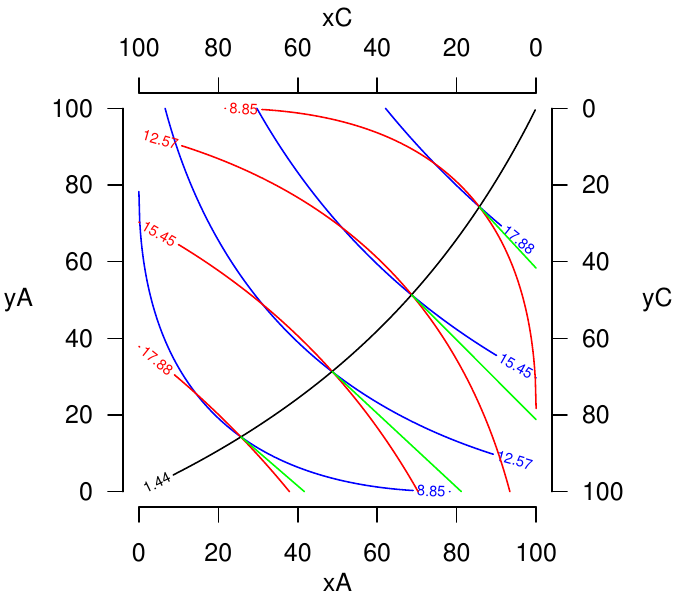}
	\caption{Revised Edgeworth box}
	\label{fig:tokensedge}
\end{figure}
\par
The green lines in the figure have slopes which are the geometric mean of the gradients of the two utility contours on the revised contract curve. We seek a point on the revised contract curve such that the corresponding green line passes through the players' starting position. Figure~\ref{fig:solutionedge} shows what we find.
\begin{figure}[H]
	\centering
	\includegraphics[scale= 1.0, clip]{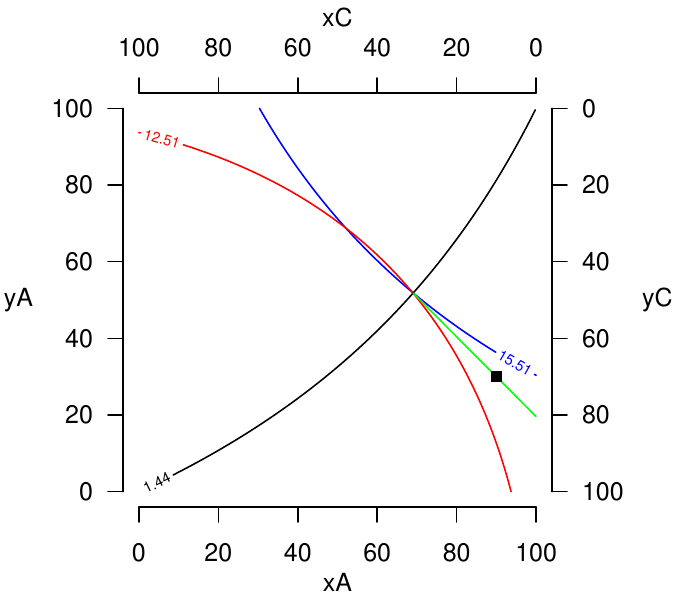}
	\caption{Solving revised Edgeworth box}
	\label{fig:solutionedge}
\end{figure}
\par
\begin{theorem}
	\label{thm:edgeworth}
	Suppose $F = (f_{1}, f_{2})$ is a point in the Edgeworth box at which the gradient of the utility curve of the first trader ($A$) is $-g$, and of the second trader ($C$) is $-h$, and such that $h/g = (1 + r)^{2}$ for some $r > 0$. Suppose that $S = (s_{1}, s_{2})$ is a starting position for the two traders with $s_{1} > f_{1}$ and $s_{2} < f_{2}$ and satisfying
	\begin{equation*}
		\frac{f_{2} - s_{2}}{f_{1} - s_{1}} = - \sqrt{gh} \text{.}
	\end{equation*}
	Then
	the market with a purchase tax rate $r$ has an equilibrium at $F$ in which the prices $p$ and $q$ for goods $1$ and $2$ (respectively) will be
	\begin{align*}
		p &= \frac{n}{r\left(s_{1} - f_{1}\right)}  \text{, and} \\
		q &= \frac{n}{r\left(f_{2} - s_{2}\right)}  \text{.}
	\end{align*}
\end{theorem}
\begin{proof}
	The suggested prices are both positive. We saw in Section that when the first trader sells good $1$ in order to buy good $2$ he trades along the line
	\begin{equation*}
		px + q(1 + r)y = n \text{,}
	\end{equation*}
	where $x$ and $y$ are the quantities of apples and pears bought (respectively). If $x = f_{1} - s_{1}$ and $y = f_{2} - s_{2}$, we find that the equation holds, so $F$ is a point on the trading line.
	\par
	It will be the optimum point if the utility contour through $F$ is tangent to the line. But we know that
	\begin{equation*}
		\frac{p}{q} = \frac{f_{2} - s_{2}}{s_{1} - f_{1}} = \sqrt{gh} \text{.}
	\end{equation*}
	We also know that $h = (1 + r)^{2} g$, so that
	\begin{equation*}
		\frac{p}{q} = (1 + r)g \text{.}
	\end{equation*}
	But the trading line has slope
	\begin{equation*}
		- \frac{p}{q(1 + r)} = -g \text{,}
	\end{equation*}
	the same as the slope of first trader's utility contour at $F$. Hence he achieves his maximum possible utility (at the assumed prices) at $F$. And similarly for the second trader. So we have found an equilibrium.
\end{proof}
\begin{corollary}
	At the equilibrium, the prices $p$ and $q$ for goods $1$ and $2$ (respectively) satisfy $p/q = \sqrt{gh}$. \qed
\end{corollary}
The second corollary would be entirely obvious in a system without tokens.
\begin{corollary}
	The equilibrium solution involves each trader paying exactly as much to buy one good as he receives from selling the other.
\end{corollary}
\begin{proof}
	At the equilibrium player $A$ receives $p(s_{1} - f_{1})$ from selling good $1$, and pays $C$ exactly the same amount ($q(f_{2} - s_{2})$) to purchase good $2$.
	\par
	So each is left with $n$ tokens after these trades. They each owe the bank $rp(s_{1} - f_{1}) = rq(f_{2} - s_{2}) = n$ tokens.
\end{proof}
\par
In our example, $F = (69.03, 51.80)$. At $F$, $g = 0.866$, and $h = 1.247$ so that $h/g = 1.44 = 1.2^{2}$. The geometric mean of the two slopes is $1.040$. This is equal to the slope between $S = (90, 30)$ and $F$, so Theorem~\ref{thm:edgeworth} applies, and there will be an equilibrium with $p = 30/20.97 = 1.43$ and $q = 30/21.80 = 1.38$. Player $A$ will sell $20.97$ apples and buy $21.80$ pears. She receives $30$ crowns for the apples, and she started with $6$ crowns. The pears cost her $30$ crowns (paid to the seller), but she must give $6$ crowns to the bank. Her budget is balanced. Similarly for trader $C$.
\par
Clearly there will be a second revised contract curve lying above the line $y = x$ where $h/g = 1/(1 + r)^{2}$. Figure~\ref{fig:solutionregion} shows the two curves as blue lines. The red lines show the corresponding curves when $r = 1/10$.
\begin{figure}[H]
	\centering
	\includegraphics[scale=1.0, clip]{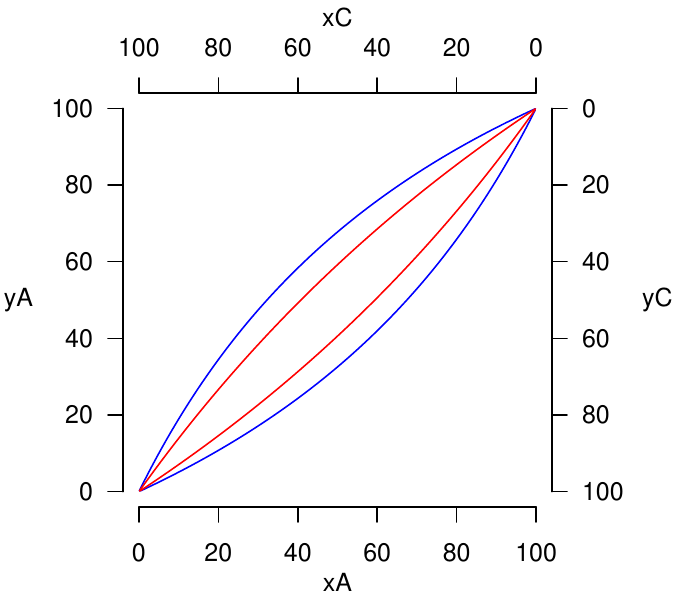}
	\caption{Solution regions for $r = 1/5$ and $r = 1/10$}
	\label{fig:solutionregion}
\end{figure}
In the lens-shaped region between the two curves there will be no equilibrium solution (or at least none with positive prices). We noted this possibility in Section~\ref{subsec:existence}.
\section{Two periods}
Suppose that the group of people is now told that in fact it will be two weeks before they are rescued. After one week there will be an airdrop of fresh supplies, again with one box of supplies for each person. The boxes will be individually assigned, so once again the traders will have property rights. They need to trade in the second period as in the first, although the goods they are trading may have changed.
\par
The real money system can easily cope with this. Keeping the same values of $n$ and $r$, a fresh supply of crowns is issued by the bank at the start of the second week. However, there is now a possible difficulty. Traders who are rich in the first week may decide to insure against being poor in the second week by not using all their crowns in the first week. This would result in less crowns being available in week one and more in week two. So the price of goods would be lower in the first week than in the second. Realising this, sensible rich traders in week one should not hoard crowns, but rather lend them to poor people, the debt to be cleared (with interest) in the second week. Then in each period the same number of crowns is used, although there is some further redistribution of wealth at the start of the second week.
\par
This process would be helped if the finance committee and the central bank announce that they plan to control prices so as to produce a slight inflation in the second period. The idea would be to keep $r$ the same but make a small increase in $n$ for week two. Everyone would then anticipate a small increase in prices, and money markets would develop allowing rich people to save money in the first period and have it repaid with interest in the second. So all the crowns issued in the first period will be used in the first period (none will be hoarded `under the bed'), and the central bank retains firm control of the price level in each period.
\par
An alternative approach would be to give tokens a limited life, so that week one tokens have no value in week two. Either way, it seems a good idea to ensure that money is never a store of value, and the use of controlled inflation seems simpler.
\par
Clearly this system can be extended to more than two time periods, or to economies with no final time period. In the latter situation, a pool of tokens circulates indefinitely (slightly increasing each time period), with the velocity of circulation of tokens between the traders and the central bank strictly controlled (in each period a token moves from the bank to a trader, and then returns). This is why prices are now determined by the quantity of tokens in the economy.
\section{Related work}
\label{sec:related}
\par
Before the rigour of general equilibrium theory, the idea that the quantity of money determines prices dates back at least to Copernicus, and was notably revived and restated in the last century by Milton Friedman.  An alternative idea that fiat money has value because the government will accept it for the payment of taxes goes back at least to Adam Smith, and is supported by \cite{Starr2013}, who relates it to general equilibrium theory and to the more recent trading post games, and gives a clear account of current work.
Even in his seminal paper \cite{Hahn1965} suggested a solution to his own problem. Suppose all traders wish to insure against not being able to make the trades they want at the auctioneer's suggested market prices. Suppose they can only insure with the government, and the government redistributes the proceeds of its insurance business to the traders. Then under certain assumptions an equilibrium will exist with money having a non-zero value. But Hahn did not see why, even supposing traders wished to take out such insurance, agents other than the government could not offer the same service. He was thinking in terms of a model rather than a mechanism.
\par
\cite{vonWeizsacker1974}, \cite{Howard1976}, and \cite{Lerner1979} (see also \cite{LernerColander1980}) all suggested introducing tradeable permits or certificates to control either price levels (von Weizs{\"a}cker) or value-added (Howard and Lerner). \cite{Layard1982} advocated a counter-inflation tax. These proposals all took a macro-economic viewpoint.
\par
At about the same time, at the micro-level, \cite{ShapleyShubik1977} introduced the idea of replacing a Walrasian market with a non-cooperative game (a Strategic Market Game) where players exchange money for goods at trading posts. There is now an extensive literature on these games. Especially relevant to our mechanism are the papers by \cite{DubeyGeanakoplos1992, DubeyGeanakoplos2003b, DubeyGeanakoplos2006}. These describe an SMG in which price levels are determinate. Traders start with commodities and `outside money', but if they need more money they can borrow `inside money' from a central bank. After using the money (the sole medium of exchange) to trade, they must repay what they have borrowed to the central bank plus an interest charge. In equilibrium, all the outside money must leave the system as interest payments, so if an interest rate is set, prices will have to adjust to make this happen. There are obvious differences between the Dubey and Geanakoplos papers and our system. They are trying to model an economy; we are suggesting a mechanism. They use trading posts; we do not. We provide new economic control variables (the number of tokens issued per head and the purchase tax rate); they do not. They insist traders use real money, of which there is a limited supply and which can be used only once in a time period; we allow traders to create their own money, and allow money to circulate.
\par
However there are also significant similarities between their system and ours (and also Hahn's). Our tokens play the same role as the outside money, and our purchase tax percentage mirrors the interest rate set in the Dubey and Geanakoplos model or the insurance premium in Hahn's story. But note that with our system, a futures market in tokens could determine an interest rate separate from the inflation rate controlled by the central bank.
\section{Summary and Conclusions}
We have devised an exchange economy in which
\begin{itemize}
	\item traders are initially supplied with `gold standard' money in the form of tokens;
	\item they then trade using either tokens or promissory notes for tokens;
	\item all trades of money for goods incur a purchase tax which must be paid in tokens.
\end{itemize}
This system gives a determinate price level without relying on any one good as num\'eraire, even in a one-period case without access to external money.
It also gives the financial authorities extra levers with which to control the economy.
\newpage
\begin{appendix}
	\renewcommand{\theequation}{A.\arabic{equation}}
	\section*{Appendix: Algorithm to calculate results}%
	\label{app:program}
	We suppose there are $k$ traders indexed by $i$. There are two goods. Trader $i$ has initial holdings of $s_{i}$ of the first good and $t_{i}$ of the second. All traders also start with $n$ tokens and have to pay a purchase tax at rate $r$. All traders have the same square root utility function: $\sqrt{x} + \sqrt{y}$ for a trader with holdings of $x$ units of the first good and $y$ of the second. The first good has price $p$ and the second $q$.
	\par
	We seek points $(x, y)$ such that the utility contour passing through the point has gradient $-g$. So starting from
	\begin{equation}
		\label{eq:constu}
		\sqrt{x} + \sqrt{y} = \text{constant,}
	\end{equation}
	we differentiate to find
	\begin{align}
		\frac{1}{2\sqrt{x}} + \frac{1}{2\sqrt{y}}\frac{dy}{dx} = 0 \text{,}
		\\
		\frac{1}{\sqrt{x}} - \frac{g}{\sqrt{y}} = 0 \text{,}
		\\
		\label{eq:gradg}
		y = g^{2} x \text{.}
	\end{align}
	\par
	Trader $i$ has three constraints (see Section~\ref{fig:tokens} and Figure~\ref{fig:tokens}). The first of these (the solid line in the figure) which we will call constraint $E$ is
	\begin{equation}
		\label{eq:constE}
		(1 + r)p\left(x - s_{i}\right) + (1 + r)q\left(y - t_{i}\right) \le n \text{.}
	\end{equation} 
	The other two constraints (the upper and lower dashed lines) which we will call $D$ and $F$ are
	\begin{align}
		\label{eq:constsDF}
		p\left(x - s_{i}\right) + (1 + r)q\left(y - t_{i}\right) &\le n \text{,}
		\\
		(1 + r)p\left(x - s_{i}\right) + q\left(y - t_{i}\right) &\le n \text{.}
	\end{align}
	So we can now calculate by using equation~\eqref{eq:gradg} the points the trader would move to if each constraint were the only restriction.
	\par
	So for constraint $E$ when $g(E) = p/q$ we find he would move to $\left(u_{i}(E), v_{i}(E)\right)$ where
	\begin{align}
		\label{eq:moveE}
		u_{i}(E) &= \frac{1}{g(E)}\left(\frac{(1 + r)\left(ps_{i} + qt_{i}\right) + n}{(1 + r)(p + q)}\right)  \text{, and}
		\\
		v_{i}(E) &= g(E)\left(\frac{(1 + r)\left(ps_{i} + qt_{i}\right) + n}{(1 + r)(p + q)}\right)  \text{.}
	\end{align}
	\par
	Similarly for constraint $D$ (when $g = p/((1 + r)q)$), or $F$ (when $g = (1 + r)p/q$), we find he would move to $\left(u_{i}(D), v_{i}(D)\right)$ or $\left(u_{i}(F), v_{i}(F)\right)$ respectively, where
	\begin{align}
		\label{eq:moveDF}
		u_{i}(D) &= \frac{1}{g(D)}\left(\frac{ps_{i} + (1 + r)qt_{i} + n}{p + (1 + r)q}\right)  \text{,}
		\\
		v_{i}(D) &= g(D)\left(\frac{ps_{i} + (1 + r)qt_{i} + n}{p + (1 + r)q}\right)\text{,}
		\\
		u_{i}(F) &= \frac{1}{g(F)}\left(\frac{(1 + r)ps_{i} +qt_{i} + n}{(1 + r)p +q}\right)  \text{,}
		\\
		v_{i}(F) &= g(F)\left(\frac{(1 + r)ps_{i} +qt_{i} + n}{(1 + r)p +q}\right)  \text{.}
	\end{align}
	\par
	The best strategy for $i$ will either be one of these three points, or else one of the two corner point solutions. Constraints $D$ and $E$ intersect at
	\begin{align}
		\label{eq:moveDE}
		u_{i}(DE) &= s_{i} \text{, and}
		\\
		v_{i}(DE) &= t_{i} + \frac{n}{(1 + r)q}  \text{.}
	\end{align}
	Similarly, $E$ and $F$ intersect at
	\begin{align}
		\label{eq:moveEF}
		u_{i}(EF) &= s_{i} + \frac{n}{(1 + r)p} \text{, and}
		\\
		v_{i}(EF) &= t_{i} \text{.}
	\end{align}
	\par
	Because we have only two goods, it is easy to see geometrically which of the five possible solutions will be feasible and optimal for player $i$. If $u_{i}(D) \le s_{i}$, then the solution will be $\left(u_{i}(D), v_{i}(D)\right)$. Similarly, if $v_{i}(F) \le t_{i}$, then the solution will be $\left(u_{i}(F), v_{i}(F)\right)$. If $u_{i}(D) > s_{i}$ but $u_{i}(E) \le s_{i}$, then the solution will be the $DE$ corner point. If $v_{i}(F) > t_{i}$ but $v_{i}(E) \le t_{i}$, then the solution will be the $EF$ corner point. Finally, if none of these apply, the solution will be $\left(u_{i}(E), v_{i}(E)\right)$. So we have the following algorithm to calculate the final holdings $\left(u_{i}(p, q), v_{i}(p, q)\right)$ of the two goods for player~$i$.
	\par
	\IncMargin{1em}
	\begin{algorithm}[H]
		\label{alg:bestoffive}
		\SetKwInOut{Input}{input}\SetKwInOut{Output}{output}
		\Input{ Prices $p$, $q$}
		\Output{ Final holdings $u_{i}$ of good $1$ and $v_{i}$ of good $2$ for $i = 1, \ldots, k$}
		\BlankLine
		\For{$i \leftarrow 1$ \KwTo $k$}{
			\If{$u_{i}(D)(p, q) < s_{i}$}{
				$u_{i} \leftarrow u_{i}(D)(p, q)$, $v_{i} \leftarrow v_{i}(D)(p, q)$\\
				\uElseIf{$v_{i}(F)(p, q) < t_{i}$}{
					$u_{i} \leftarrow u_{i}(F)(p, q)$, $v_{i} \leftarrow v_{i}(F)(p, q)$\\
					\uElseIf{$u_{i}(E) < s_{i}$}{
						$u_{i} \leftarrow u_{i}(DE)(p, q)$, $v_{i} \leftarrow v_{i}(DE)(p, q)$\\
						\uElseIf{$v_{i}(E) < t_{i}$}{
							$u_{i} \leftarrow u_{i}(EF)(p, q)$, $v_{i} \leftarrow v_{i}(EF)(p, q)$\\
							\Else{
								$u_{i} \leftarrow u_{i}(E)(p, q)$, $v_{i} \leftarrow v_{i}(E)(p, q)$
		}}}}}}
		\caption{Calculating $u_{i}$ and $v_{i}$}
	\end{algorithm}
	\DecMargin{1em}
	\par
	Hence for given prices we can find the demand for each good from each trader. Summing over traders we find the total excess demand (positive or negative) for each good for prices $p$ and $q$.
	\par
	We can now plot the two lines where demand equals supply (see Figure~\ref{fig:threecontours}), and solve numerically for the crossing point(s). Alternatively we can minimise the sum of squares of the excess demands. Hence we find solution prices. We can also see which of the five solution points each player uses at the solution, and this may lead to a closed form solution: thus in our example the solution value for $q$ is the root of a cubic.
	\par
	The results given in the paper are now easily calculated. For the taxation example (Subsection~\ref{subsec:taxation}), we add the government as a fourth player. We must also modify strategy $E$, as indicated in the subsection, and recalculate the corner solutions $DE$ and $EF$. Algorithm~\ref{alg:bestoffive} must also be adjusted: for example, strategy $D$ is now chosen when $v_{i}(D) \ge t_{i}$. For the Edgeworth box example (Section~\ref{sec:edgeworth}), we simply reduce the number ($k$) of players to $2$.
\end{appendix}
\newpage
\singlespacing
\bibliographystyle{ecta}            %
\bibliography{Exchange}
 \end{document}